\documentclass[12pt]{amsart}
\usepackage{graphicx}
\newtheorem{theorem}{Theorem}
\newtheorem{definition}[theorem]{Definition}
\begin{document}

\centerline{\bf \Large Equilibrium distributions}
\centerline{\bf \Large and superconductivity}
\vskip 1cm

\centerline{\bf Ashot Vagharshakyan}

\vskip 0.5cm
\centerline{Institute of Mathematics National Academy of sciences of Armenia}
\vskip 0.2cm
\centerline{vagharshakyan@yahoo.com}
\vskip1cm
{\it Abstract. In this article two models for charges distributions are discussed. 
On the basis of our consideration we put different points of view for stationary state. 
We prove that only finite energy model for charges distribution and well known variation principle explain some well known experimental results.

	A new model for superconducting was suggested, too. In frame of that model 
some characteristic experimental results for superconductors are possible to  
explain.}
\vskip 0.5cm


\section{Introduction}

	In all metals there are particles caring positive and negative charges. A body appears us neutral. Because of the positive and negative charged particles are accurately balanced.

	Note that we do not know and, moreover, we can not take into account the all possible 
interactions between elementary particles. So, the construction of a comprehensive   
model for metal is meaningless. That model can not present practical interest for its 
mathematical complexity too. Therefore, one prefers to neglect number of details and to
build a model as simple as possible. 

	In framework of each model, we expect to explain a given number of experimental results. 
Each model has a limited capacity, and ceases to be true outside of those limits.

	We consider two types of charges distribution. 

First the point charges model is. Second the finite energy model is.
 
	We show that the point charges model bring us to contradictions, with some experimental 
results. So, in spite of its simplicity and facility we need to reject it.

	We discuss the principles, which characterize the stationary states, too.

	Two kinds of stationary sates we consider. 
First the equilibrium state is and the second the static state is.
 
	At the equilibrium state, the charges must be distributed in such a way, that the 
potential energy of the whole system reaches its minimal value.

	In static state the charges must be distributed in such a way, that:
 
1.	The force acting to each charge, placed inside of conductor, equals zero;

2.	The force acting to a charge is directed out of conductor 
if it is placed on the boundary of conductor.

	Let us note that the static states bring us to a contradiction with experiment too. Thus we come to the equilibrium state.

	We prove that in the finite energy model the distribution with minimal energy exists. 
It is unique and stable.

	In addition, for equilibrium distribution, we prove that the corresponding potential 
function is constant inside on each component of conductor.  The last result explains 
Cavendish's experiment.
  
	There are well known BCS model, which explain the superconducting phenomenon, 
see \cite{b:Sh}. 
On the bases of BCS theory the effect of appearance the attraction between two electrons, in low 
temperature, is. That effect is conditioned by crystal's nodes specific oscillations. 

	We give a new model based on classical electrodynamics. In frame of new model it is 
possible to explain some experimental results for superconductors. 

	In our model we discuss and explain the following experimental results. 

	Experimentally was detected that for some metals, at very low temperature, 
the resistance suddenly falls to zero. 
Now, this phenomenon is known as superconductivity.

	The effect before the critical temperature. 

	If the temperature decreases the resistance decreases too. 
For some materials near the critical temperature the resistance 
increases a little and after reaching some maximum 
value quickly droppers to zero.

	For some metals the superconductivity 
property was observed only in huge pressure.
 
	It was verified that superconductivity is destroyed in presence of  
sufficiently strong magnetic field.

	Superconductivity is destroyed also, 
when the current becomes greater of some critical value.

	Some ideal conductors are bad superconductors and vice-versa.

\section{ The point charges model}

In point charges model we assume that charges are located at points and they have no 
inside structure.

	Certainty, this point of view is primitive. In creating this model we take into 
account the following experimental results.

	The distance, between two elementary particles, is bigger $10^{-8}$ cm. The experimental 
results show, that the size of each elementary particle is less  $10^{-12}$   cm. 
So, the elementary particles are placed faraway, in compared of them geometrical sizes. 
This fact benefits the point charges model.

	In 1785 Coulomb proposed an experiment to measure the force of interaction between 
two small charges. The experimental results give, that the force is inversely proportional 
to the square of charge's distance.

	Coulomb's experiment shows, that the total force of a number charges to a given one 
charge equals vector sum of the forces between pair of charges. This last fact is known as 
superposition principle.

	It is important to emphasize that Coulomb's law is valid only for the bodies which have 
small geometric sizes in compared with the distances between them. Note that only in that 
case "the distance between two charges" has a meaning.

	There are two areas, where there is no firm belief that Coulomb's law is valid. The 
distances less  $10^{-14}$ cm., where nucleus forces dominate and the distance larger several 
kilometers, where the immediate experimental measurements are considerably hamper.

	In 1910 Milliken, by developing Erengaft's method, measured electron's charge and he founded, 
the value  $e=-1,60 10^{-19}$ Coulomb.

	In 1919 Aston, by mass - spectrograph method found, that the total mass of each atom equals 
to an entire multiple of some fixed quantity. This fact permits to conclude, that nucleus 
consists with elementary particles with the same mass. Now it is known, that some of those 
particles carry positive charges and some others don't carry any charge. The positive charged
particles are named protons. It is known that proton carry minimum positive charge equals $-e$.

	It is very important to note, that the electrical forces are very strong in compared with 
the gravitation forces. For example, the fraction of electrical push force $F_e$ of two protons, 
to their gravitational attraction force  $F_g$, equals
\vskip -0.8cm
\begin{equation*} 
\frac{F_e}{F_g}=1,24 10^{36}.
\end{equation*}
\vskip -0.3cm
That is way, we disregard gravitation forces in our consideration.

	Thus there are elementary particles, which carry minimal positive and negative charges. 
Note that there are other particles carry charges too. For example mesons.

	In despite of considerable experimental efforts, particles with fractional to $e$ charges 
are not detected.

\section{The forces in point charges model}

	Let a point charge $q_0$ be at the fixed point $\vec{x}_0$, and another charge $q$ is at the
 point $\vec{x}$. By Coulomb's law, on the second charge acts the force
\begin{equation*}
\vec{F}(\vec{x})=q_0q\frac{\vec{x}-\vec{x}_0}{||\vec{x}-\vec{x}_0||^3}=
-q_0q\nabla\left(\frac{1}{||\vec{x}-\vec{x}_0||}\right).
\end{equation*}

	Now let we have a point charges $q_j,\,\, j=1,2,\dots,n$ placed at the points 
$\vec{x}_j,\,\, j=1,2,\dots,n$. We can present this 
distribution of charges as a generalized function
\begin{equation*}
l=\sum_{k=1}^nq_k\delta(\vec{x}-\vec{x}_k),
\end{equation*}
where $\delta(\vec{x})$ is Dirak's function.
The vectors
\begin{equation*}
\vec{F}_k=\sum_{j=1, j\neq k}^n q_jq_k\frac{\vec{x_k}-\vec{x}_j}{||\vec{x}_k-\vec{x}_j||^3},\quad k=1,2,\dots,n.
\end{equation*}
are the forces, acting on the charge placed at the points $\vec{x}_k$, by others.

	For an arbitrary
\begin{equation*}
0<2r<\min_{i\neq j}||\vec{x}_i-\vec{x}_j||
\end{equation*}
we define new generalized functions $l_r$ acting on an infinitely differentiable testing function $\varphi(\vec{x})$ with compact support, as follows:
\begin{equation*}
l_r(\varphi)=\sum_{k=1}^n \frac{q_k}{4\pi r^2}\int_{\partial B(\vec{x}_k,r)}\varphi(\vec{x})d m_2(\vec{x}).
\end{equation*}
It is obvious, that
\begin{equation*}
\lim_{r\to 0}l_r(\varphi)=\sum_{k=1}^n q_k\varphi(\vec{x}_k)=l(\varphi).
\end{equation*}
Roughly speaking we build $l_r$ by uniformly spreading over the 
sphere $\partial B(\vec{x}_k,r)$, the point charge $q_k$ placed at 
the point $\vec{x}_k$. Since
\begin{equation*}
\frac{1}{4\pi r^2}\int_{\partial B(\vec{x}_k,r)}\frac{d m_2(\vec{x})}{||\vec{x}-\vec{y}||}=
\min\left(\frac{1}{||\vec{x}-\vec{x}_k||},\frac{1}{r}\right),
\end{equation*}
so, we have
\begin{equation*}
U^{l_r}(\vec{x})=l_r\left(\frac{1}{||\vec{x}-\vec{y}||}\right)=
\sum_{k=1}^nq_k\min\left(\frac{1}{||\vec{x}-\vec{x}_k||},\,\,\frac{1}{r}\right).
\end{equation*}
	For an arbitrary unit vector $\vec{n}$ we have
\begin{equation*}
l_r\left(U^{l_r}(\vec{x})\frac{\partial\varphi(\vec{x})}{\partial\vec{n}}\right)=\lim_{t\to 0}l_r\left(\frac{\varphi(\vec{x}+t\vec{n})-\varphi(\vec{x})}{t}U^{l_r}(\vec{x})\right)
=
\end{equation*}
\begin{equation*}
=-\lim_{t\to 0}l_r\left(\frac{U^{l_r}(\vec{x}-t\vec{n})-U^{l_r}(\vec{x})}{-t}\varphi(\vec{x})\right)=-l_r\left(\frac{\partial U^{l_r}(\vec{x})}{\partial\vec{n}}\varphi(\vec{x})\right)=
\end{equation*}
\begin{equation*}
=\sum_{k\neq j=1}^nq_kq_j
\frac{(\vec{x}_k-\vec{x}_j,\,\,\vec{n})}{||\vec{x}_k-\vec{x}_j||^3}\varphi(\vec{x}_k)=\sum_{k=1}^n\varphi(\vec{x}_k)(\vec{F}_k,\,\,\vec{n}).
\end{equation*}
	This formula make natural to define forces as a generalized function
\[
\vec{F}(\varphi)=l_r\left(U^{l_r}\nabla\varphi\right).
\]
In point charges model we can not pass to the limit if $r\to 0$ in this formula. 

However, we will see later, that in finite energy model that limit exists. 

\section{Conductor}

	In this section we introduce the conception of conductor. 

We postulate some important properties of conductors without discussing the internal mechanisms 
of their appearance.

	At first let us introduce some definitions.
\begin{definition}
We say that $\vec{x}_0\in E$ is an inner point of a set $E$ if there is 
a $r>0$ such that

\begin{equation*}
B(\vec{x}_0,r)=\{\vec{x};\,\,\, ||\vec{x}-\vec{x}_0||<r\}\subset E.
\end{equation*}
\end{definition}
	The number of all inner points we denote by $\dot E$. The set $E$ is open if  $\dot E=E$.

\begin{definition} We say that $\vec{x}_0$ is a boundary point of a subset $E$ if for 
each $r>0$ we have
\begin{equation*}
B(\vec{x}_0,r)\cap E\neq \emptyset,\quad E\setminus B(\vec{x}_0,r)\neq \emptyset.
\end{equation*}
\end{definition}

\begin{definition}
Let $\vec{x}_0\in E$. If for a nonzero vector $\vec{n}$ we have
\begin{equation*}
\lim_{E\ni \vec{x}\to \vec{x}_0}\frac{\left(\vec{x}-\vec{x}_0, \vec{n}\right)}{||\vec{x}-\vec{x}_0||}=0
\end{equation*}
then $\vec{n}$ is named a normal to $E$, at the point $\vec{x}_0$.
\end{definition}
\begin{definition}
We say that a unit vector $\vec{n}$ has an inner direction for a 
subset $E$, at the 
point $\vec{x}_0\in E$, if for an arbitrary $\epsilon>0$ there are $0<r<t<\epsilon$ such that 
\begin{equation*}
B(\vec{x}_0+t\vec{n},r)\subset E.
\end{equation*}  
	Note that if $\vec{x}_0\in \dot E$ then an arbitrary unit vector $\vec{n}$ has an inner 
direction.
\end{definition}
\begin{definition}
The subset $E$ is conductor, if

1.	The charges can freely move through $E$;

2.	There are some forces keeping charges inside of $E$;

\end{definition}
	It is important to emphasize, that if a conductor has several connected components, 
then charges can not jump from one component to other.
	The forces, which keep the charges inside of conductor cannot have electrical genesis. 
This follows that an electrostatic problem for conductors, does not possible to solve using 
only electrical forces. 

	We guess that there are other forces, too. 
Here we use only the consequence of those forces existence postulated in the point 2.

\section{Static state in the point charges model}

\begin{definition}
Let $E$ be a conductor. Families of point charges are in static state if

1.	On a charge lies inside of conductor, by other charges, act a force equals zero;

2.	On a charge placed on the conductor's boundary acts a force, which cannot move it 
to an inside direction of the conductor.
\end{definition}
	If a conductor has smooth boundary, then the last condition means that on the charge 
placed on the boundary by other charges act a force, which is normal to the boundary 
and it is directed outside of the conductor's boundary.
  
	If the boundary $\partial E$ is smooth, then those conditions we can write in the 
following form
\begin{equation*}
\sum_{j\neq i=1}^n\frac{q_i\left(\vec{x}_i-\vec{x}_j\right)}{||\vec{x}_i-\vec{x}_j||}=0,\quad \vec{x}_j\in \dot E.
\end{equation*}
and
\begin{equation*}
\sum_{j\neq i=1}^n\frac{q_i\left(\vec{x}_i-\vec{x}_j,\,\vec{n}_j\right)}{||\vec{x}_i-\vec{x}_j||}\vec{n}_j=\sum_{j\neq i=1}^n\frac{q_i\left(\vec{x}_i-\vec{x}_j\right)}{||\vec{x}_i-\vec{x}_j||},
,\quad \vec{x}_j\in \partial E.
\end{equation*}
where $\vec{n}_j$ is the unit outer normal to the boundary $\partial E$ at the 
point $\vec{x}_j$.
 
	Note that in definition of static state, we do not put an additional condition, that 
on each component of conductor there are only the same sign charges. 

	The following well known example one can find in \cite{b: T}. Let three charges lie on a segment. 
Two of them are at the ends and have the same charge 
equal $4q$ and the third one is at the middle of the segment and its charge is $-q$.  
It is easy to check, that the force on each charge acting by two others, equals zero.

	Note that, the static state for the ball is not unique. Indeed, if we roll the ball 
around an axis passing through the center, all the charges will remain inside of the ball 
and the forces will not change.

	Note that the uniqueness of static state depends upon the shape of a conductor. 
For example, the four equal charges lie on the vortexes of a tetrahedron, form the unique 
stationary state. 

\section{Equilibrium distribution in the point charges model}

	Let a particle move by the path 
\begin{equation*}
\gamma=\{\vec{x}(t);\,\, a\leq t\leq b\}
\end{equation*} 
in the field of forces $\vec{F}(\vec{x})$. Then the work done by that partical equals
\begin{equation*}
\int_{\gamma}\left(\vec{F}(\vec{x}),d \vec{x} \right)=
\int_a^b\left(\vec{F}(\vec{x(t)}),\vec{\dot{x}}\right)dt.
\end{equation*}
	Let a charge $q_0$  be at the point $\vec{x}_0$ , and another charge $q$ be at $\vec{x}$.
Suppose the point charge $q$ slowly goes to infinity. It is well known, that the work 
done by that motion, regardless of a path form, equals
\begin{equation*}
\frac{q_0 q}{||\vec{x}_0-\vec{x}||}
\end{equation*}
	This observation makes natural to determine the potential energy of 
$q_j,\quad j=1,2,\dots,n$ placed at the points $\vec{x}_j,\quad j=1,2,\dots,n$ by the 
following formula 
\begin{equation*}
W=\sum_{1\leq i<j\leq n}\frac{q_i q_j}{||\vec{x}_i-\vec{x}_j||}
\end{equation*}

	Let us note that potential energy for a family of point charges may be positive negative or zero.

	To consider the equilibrium distributions we need to put the following 
additional condition on conductor:

3.	On each component of conductor can be charges only of the same sign. 

	Note that, if we have two particles with the positive and the negative charges, 
then if we bring them nearer, the potential energy can take an arbitrarily negative value.

	The condition 3 excludes the above mentioned unwanted effect.

\begin{definition} 
We say that the given finite number point charges placed on $E$, 
are in equilibrium state, if they potential energy takes minimal value among the 
all possible distributions.
\end{definition}

	This condition we can write in the following form.
Let the compact subset $E$ consists of the finite number disjoint connected components, i.e.
\[
E=\bigcup_{j=1}^{n}E_j
\]
Let the point charges $\left\{q_{kj}\right\}_{k=1}^{m_j}$ are placed in $E_j$. 
If they are situated at the points
$\left\{\vec{x}_{kj}\right\}_{k=1}^{m_j}$ then they are in equilibrium state if

\begin{equation*}
\sum_{j=1}^n\sum_{1\leq p<q\leq m_j}\frac{q_{pj}q_{qj}}{||\vec{x}_{pj}-\vec{x}_{qj}||}
+\sum_{1\leq i<j\leq n}\sum_{1\leq p \leq m_i}\sum_{1\leq q\leq m_j}\frac{q_{pi}q_{qj}}{||\vec{x}_{pi}-\vec{x}_{qj}||}=
\end{equation*}
\begin{equation*}
=\inf_{\vec{y}_{kj}\in E_j}
\left(\sum_{j=1}^n\sum_{1\leq p<q\leq m_j}\frac{q_{pj}q_{qj}}{||\vec{y}_{pj}-\vec{y}_{qj}||}
+\sum_{1\leq i<j\leq n}\sum_{1\leq p \leq m_i}\sum_{1\leq q\leq m_j}\frac{q_{pi}q_{qj}}{||\vec{y}_{pi}-\vec{y}_{qj}||}\right).
\end{equation*}

\begin{theorem} Let $E$ be a compact subset. Let inside of each component of conductor $E$ 
we have finite number point charges with the same sign. Then the equilibrium distribution 
always exists.
\end{theorem}
\begin{proof} 
We have the inequalities $q_{pj}q_{qj}>0$ if $1\leq p<q\leq m_j$.
The potential energy is continuous from bellow. 
So, it reaches minimum value. 
\end{proof}

\begin{theorem}
In equilibrium state, all charges must be on conductor's boundary. 
\end{theorem}
\begin{proof} 
Suppose that in the equilibrium state a charge, say $\vec{x}_1$, lies 
inside of the conductor $E$. This follows that there is a ball $B(\vec{x}_1,\,r)\subset E$ 
satisfying the condition 
\begin{equation*}
\{\vec{x}_1,\dots,\vec{x}_1\}\cap B(\vec{x}_1,\,r)=\emptyset.
\end{equation*}
Since the function
\begin{equation*}
V(\vec{x})=\sum_{j=2}^n\frac{q_1q_j}{||\vec{x}-\vec{x}_j||}+
\sum_{i=2,i\neq j}^n\frac{q_iq_j}{||\vec{x}_i-\vec{x}_j||},
\quad \vec{x}\in R^3\setminus \{\vec{x}_1,\dots,\vec{x}_n\},
\end{equation*}
is harmonic so, by mean value principle, we have
\begin{equation*}
V(\vec{x}_1)=\sum_{i=1,i\neq j}^n\frac{q_iq_j}{||\vec{x}_i-\vec{x}_j||}=
\frac{1}{4\pi r^2}\int_{\partial B(\vec{x}_1,r)}V(\vec{y})dm_2(\vec{y}).
\end{equation*}
	The function $V(\vec{x})$ can not be constant on the sphere $\partial B(\vec{x}_1,r)$, 
and therefore there is a point $\vec{x}_0\in \partial B(\vec{x}_1,r)$ such that 
\begin{equation*}
V(\vec{x}_0)<V(\vec{x}_1).
\end{equation*}
	Consequently, there is another charge's distribution with a lower potential energy.  
\end{proof}
 
	{\bf Example 1.} There is a static state of the same sigh charges, which is different of 
equilibrium state. 

	Indeed, let we have the same charges $Q>0$ placed on the vertices 
\begin{equation*}
\left(1,\,\,0\,\,,0\right),\quad \left(-\frac{1}{2},\,\,\frac{\sqrt 3}{2}\,\,,0\right),
\quad \left(-\frac{1}{2},\,\,-\frac{\sqrt 3}{2}\,\,,0\right),\quad \left(0,\,\,0\,\,,\sqrt 2\right). 
\end{equation*}
of tetrahedron and one charge $Q$ placed at the center point
\begin{equation*}
\left(0,\,\,0\,\,,\frac{\sqrt 2}{4}\right).
\end{equation*}
The force acting on the last charge equals zero. This distribution cannot be 
equilibrium since one charge is placed inside of tetrahedron.

\begin{theorem} 
Let $E$ be a compact and convex subset with smooth boundary. 
Then, in equilibrium state, the force acting on each charge by others is perpendicular 
to the boundary   and it is directed outer of $E$.
\end{theorem} 
\begin{proof} 
Define the function
\begin{equation*}
\Phi(\vec{x})=-d(\vec{x},\partial E),\quad \vec{x} \in E,
\end{equation*}
and
\begin{equation*}
\Phi(\vec{x})=d(\vec{x},\partial E),\quad \vec{x}\in R^3\setminus E,
\end{equation*}
where 
\begin{equation*}
d(\vec{x}, \partial E)=\inf_{\vec{y}\in \partial E}||\vec{x}-\vec{y}||.
\end{equation*}
Let $\vec{n}(\vec{x}_0),\,\, \vec{x}_0\in \partial E,$ 
be an outer normal to the boundary $\partial E$. We have
\[
\Phi(\vec{x}_0+t\vec{n}(\vec{x}_0))=t+o(t),\quad t\to 0.
\]

	Consequently, for each boundary point $\vec{x}_0 \in \partial E$ we 
have \[\vec{n}(\vec{x}_0)=\nabla \Phi(\vec{x}_0).\]

	Our problem to find the equilibrium state we can formulate as follows:
\begin{equation*}
\inf_{\vec{x}_k}\left\{\sum_{1\leq i<j\leq n}\frac{q_iq_j}{||\vec{x}_i-\vec{x}_j||};\,\,\,
\Phi(\vec{x}_k)=0,\,\,\, k=1,2,\dots,n\right\}.
\end{equation*}
Let us solve this extreme - value problem by Lagrange method. Denote by $G$ the auxiliary function
\begin{equation*}
G(\vec{x}_1,\dots,\vec{x}_n,\lambda_1,\dots,\lambda_n)=\sum_{1\leq i<j\leq n}\frac{q_iq_j}{||\vec{x}_i-\vec{x}_j||}+\sum_{j=1}^n\lambda_j\Phi(\vec{x}_j).
\end{equation*}
If the potential energy reaches its minimal value at the points 
\begin{equation*}
\vec{x}_k\in E,\quad k=1,\dots,n,
\end{equation*} 
then for each $k=1,\dots,n$ must be
\begin{equation*}
0=\frac{\partial G}{\partial \vec{x}_k}=-\sum_{j=1}^{k-1}q_kq_j\frac{\vec{x}_k-\vec{x}_j}{||\vec{x}_k-\vec{x}_j||^3}-
\sum_{j=k+1}^{n}q_kq_j\frac{\vec{x}_k-\vec{x}_j}{||\vec{x}_k-\vec{x}_j||^3}+
\lambda_k\vec{n}(\vec{x}_k).
\end{equation*}
Calculating the dot product of these equalities with the vectors 
\begin{equation*}
\vec{n}(\vec{x}_k),\quad k=1,\dots,n
\end{equation*} 
for each $k=1,\dots,n$ we obtain 
\begin{equation*}
q_k\left(\sum_{j=1}^{k-1}q_j\frac{(\vec{x}_k-\vec{x}_j,\vec{n}(\vec{x}_k))}{||\vec{x}_k-\vec{x}_j||^3}+
\sum_{j=k+1}^{n}q_j\frac{(\vec{x}_k-\vec{x}_j,\vec{n}(\vec{x}_k))}{||\vec{x}_k-\vec{x}_j||^3}\right)=
\lambda_k.
\end{equation*}
Since the charges $q_j$ have the same sigh and due to convexity of conductor $E$ we have 
\begin{equation*}
(\vec{x}_k-\vec{x}_j,\vec{n}(\vec{x}_k))\geq 0,
\end{equation*}
so $\lambda_k\geq 0,\quad k=1,\dots,n$. The force, acting on the charge $q_k$ by other 
particles, equals 
\begin{equation*}
q_k\left(\sum_{j=1}^{k-1}q_j\frac{\vec{x}_k-\vec{x}_j}{||\vec{x}_k-\vec{x}_j||^3}+
\sum_{j=k+1}^{n}q_j\frac{\vec{x}_k-\vec{x}_j}{||\vec{x}_k-\vec{x}_j||^3}\right)=\lambda_k\vec{n}(\vec{x}_k).
\end{equation*}
\end{proof} 
	The last theorem explains why the charges are immobile in equilibrium state.

\section{Difficulties in point charges model}

	First note, that similarly with static state the equilibrium distribution is not unique too. 
Indeed, if we turn the ball $B(\vec{0},R)$ around an axes passes through the origin, 
we get a new equilibrium state again. 

	Second, the point charges model gives the result, which contradicts Cavendish's 
experiment. Indeed, let us put on the ball $B(\vec{0},R)$, $n$ positive point 
charges $q_0=q_0(n)$. Let us assume, that the total charge $q=nq_0$ is 
constant and does not depend upon $n$. Besides of those charges, let us put an 
immobile positive charge $Q$ at the point $\vec{x}_0$, 
outside of the ball $B(\vec{0},R)$, i. e. $d=||\vec{x}_0||>R$.

	We know, that at the equilibrium state, all charges will be on the boundary 
$\partial B(\vec{0},R)$. Let us assume, that they are placed at the points
\begin{equation*}
\{\vec{x}_1,\dots,\vec{x}_n\}\subset \partial B(\vec{0},R).
\end{equation*}
These charges minimize the potential energy 
\begin{equation*}
\sum_{j=1}^n\frac{q_0Q}{||\vec{x}_0-\vec{x}_j||}+
\sum_{1\leq i<j\leq n}\frac{q_0^2}{||\vec{x}_i-\vec{x}_j||}.
\end{equation*}
Since the potential function of the whole system, is 
\begin{equation*}
U(\vec{x})=\frac{Q}{||\vec{x}_0-\vec{x}||}+
\sum_{j=1}^n\frac{q_0}{||\vec{x}-\vec{x}_j||}.
\end{equation*}
so, 
\begin{equation*}
\nabla U(\vec{x})=-Q\frac{\vec{x}-\vec{x}_0}{||\vec{x}-\vec{x}_0||^3}-
q_0\sum_{j=1}^n\frac{\vec{x}-\vec{x}_j}{||\vec{x}-\vec{x}_j||^3}.
\end{equation*}
Thus we have
\begin{equation*}
||\nabla U(\vec{0})||\geq \frac{Q}{||\vec{x}_0||^2}-
\sum_{j=1}^n\frac{q_0}{||\vec{x}_j||^2}=\frac{Q}{d^2}-\frac{q}{R^2}.
\end{equation*}
If 
\begin{equation*}
\frac{Q}{d^2}>\frac{q}{R^2}
\end{equation*}
we can not say, that the electrical field, inside of the ball $B(\vec{0},R)$, can be made 
arbitrary small, even by increasing the number $n$ of charges.

	However, the Cavendish's experiment shows, that on a small charge, placed at the center 
of the ball $B(\vec{0},R)$, practically acts not force.

		In electrostatics there is a following fenomenon. Let $E$ be a bounded
conductor in $R^3$ with smooth boundary. Consider a positive charge 1, say in the most stabie equilibrium. That is the potential energy should be minimal. Then the electrostatic potentiaal is constant throughout its interior.

	The following question naturaly arises: can we explain this fenomenon if on the conductor there
are $N$ point charges? The answere is negative becous out of charges the potential function 
is harmonic and if it is constant in some open set then it must be constant everywhere.

\section{The finite energy model}

	In the finite energy model we assume that inside of a metal, the positive charged 
nodes and the cloud of free electrons have approximately the same density. Due to 
the external influences, those densities can change. We interpret those changes 
as a simultaneous appearance of positive and negative charges.

	Let us note that a measure is the most convenient and intuitively transparent mathematical tool, to describe 
charge's distribution, see [9]. However, we will go further and we will describe 
charge's distribution using generalized functions. On mathematical point of view 
there is no problem, but on the intuitive level, this approach brings us to some 
complexities. For example, in this model the total charge, inside of a given ball, 
does not possible to determine. To this problem we will return later, too.

	Like to the point charges model, here we also assume, that there are some, not 
electrical forces, which keep charges inside of a metal.

	In this new model we define the static state and the equilibrium distribution.
 
	We will prove that the static state brings us to some results contradict experiments.

	However, the equilibrium distribution exists and it is unique. Moreover, 
the properties of equilibrium distribution explain Cavendish's experiment, too. 

\begin{definition}
A function $f(\vec{x})\in L^2(\Omega),$ belongs to the Dirichlet
space $D=D(\Omega)$ if
\begin{equation*}
||f||^2=\int_{\Omega}|f(\vec{x})|^2dm_3(\vec{x})+\int_{\Omega}|\nabla f(\vec{x})|^2dm_3(\vec{x})<\infty.
\end{equation*}
\end{definition}
\begin{definition} A function $f(\vec{x})\in L^2(\Omega),$ 
belongs to subspace $\dot{D}(\Omega)$ if there are functions 
\begin{equation*}
f_n\in D(\Omega),\quad n=1,2,\dots
\end{equation*} 
for which
\begin{equation*}
\text supp (f_n)\subset \Omega,\quad n=1,2,\dots
\end{equation*}
and
\begin{equation*}
\lim_{n\to \infty}||f-f_n||=0.
\end{equation*}
\end{definition}

\begin{definition} 
The dot product of two functions from $D$ we can define in 
two manners. First by the following formula
\begin{equation*}
(f,g)=\int_{\Omega}(\nabla f(\vec{x}),\nabla g(\vec{x}))dm_3(\vec{x}).
\end{equation*}
where $dm_3(\vec{x})$ is volume differential in $R^3$ and second by
\begin{equation*}
(f,g)_0=\int_{\Omega}f(\vec{x})g(\vec{x})dm_3(\vec{x})+\int_{\Omega}(\nabla f(\vec{x}),\nabla g(\vec{x}))dm_3(\vec{x}).
\end{equation*}
\end{definition}

\begin{theorem}
The space $D$ coincides with the set of functions 
$f(\vec{x})\in L^2(R^3)$ Fourier 
transform $\hat{f}(\vec{x})$ of which satisfy the condition: 
\begin{equation*}
\int_{R^3}|f(\vec{x})|^2(1+4\pi^2)||\hat{f}(\vec{x})|^2)dm_3(\vec{x})<\infty.
\end{equation*}
\end{theorem}
\begin {proof} 
By Parseval's equality we have: 
\begin{equation*}
\int_{R^3}|\nabla f(\vec{x})|^2dm_3(\vec{x})=4\pi^2\int_{R^3}|\hat{f}(\vec{x})|^2dm_3(\vec{x}).
\end{equation*}
\end {proof}
\begin{theorem} 
For an arbitrary bounded functions $f(\vec{x}),\,\,g(\vec{x})\in D$ 
we have $f(\vec{x})g(\vec{x})\in D$.
\end {theorem}
\begin{definition} 
Denote by $D^*$ the space of linear continuous functional, defined on $D$. 
\end{definition}
	The norm of a generalized function $l\in D^*$ is
\begin{equation*}
||l||=\sup \{|l(f)|;\,\,\, f\in D,\,\, ||f||_D\leq 1\}.
\end{equation*}
\begin{definition}
In the finite energy model, the distribution of charges may be 
an arbitrary generalized function $l\in D^*$.
\end{definition}
\begin{definition}
Let $l\in D^*$. Its potential function $U^l$ is an element from $D$, 
which present the functional $l$ by scalar product, i.e.
\begin{equation*}
l(\varphi)=(U^l,\,\,\varphi),\quad \varphi(\vec{x})\in D.
\end{equation*}
\end{definition}
	The existence of $U^l\in D$ is guaranteed by the well known M. Riesz's theorem.

	Let us note that if 
\[
l(\varphi)=l_f(\varphi)=\int_{R^3}f(\vec{x})\varphi(\vec{x})dm_3(\vec{x}),
\]
where $f\in L_1(R^3)$, then we have
\begin{equation*}
U^{l_f}(\vec{x})=\int_{R^3}\frac{f(\vec{y})}{||\vec{x}-\vec{y}||}dm_3(\vec{y}),\quad \vec{x}\in R^3.
\end{equation*}
\begin{definition}
Let $l\in D^*$. Then this functional we can present in the following form too
\begin{equation*}
l(\varphi)=(U_0^l,\,\,\varphi)_0,\quad \varphi(\vec{x})\in D.
\end{equation*}
where  $U_0^l$ is an element from $D$.
\end{definition} 
	Let us note that if $l=l_f$ then we have
\begin{equation*}
U_0^{l_f}(\vec{x})=\int_{R^3}\frac{\exp\{-||\vec{x}-\vec{y}||\}}{||\vec{x}-\vec{y}||}f(\vec{y})d_3(\vec{y}),\quad \vec{x}\in R^3.
\end{equation*}
\begin{definition}
The potential energy of the charges distribution $l\in D^*$ with a 
compact support, we define by the formula 
\begin{equation*}
W(l)=l(U^l)
\end{equation*}
and correspond by the formula
\begin{equation*}
W_0(l)=l(U_0^l).
\end{equation*}
\end{definition}
	Let us note that in the finite energy model the potential energy for an arbitrary distribution 
is a positive.
\begin{definition}
Scalar product in the space of generalized functions $D^*$, 
we define by the following formulas 
\begin{equation*}
[l_1,l_2]=l_1(U^{l_2}),
\end{equation*}
or
\begin{equation*}
[l_1,l_2]_0=l_1(U_0^{l_2}).
\end{equation*}
\end{definition}
	Thus for continuous differentiable functions $f,\,\,g$, with compact supports, 
we have
\begin{equation*}
[l_f,l_g]=l_g(U^f)=\int_{R^3}\int_{R^3}\frac{f(\vec{x})g(\vec{y})}{||\vec{x}-\vec{y}||}dm_3(\vec{x})dm_3(\vec{y})
\end{equation*}
and
\begin{equation*}
[l_f,l_g]_0=l_g(U_0^f)=\int_{R^3}\int_{R^3}\frac{\exp\{-||\vec{x}-\vec{y}||\}}{||\vec{x}-\vec{y}||}
f(\vec{x})g(\vec{y})dm_3(\vec{x})dm_3(\vec{y}).
\end{equation*}
\begin{theorem}
	If $l_1,\,\,l_2\in D^*$, then 
\begin{equation*}
[l_2,l_2]=\int_{R^3}\hat{l}_1(\vec{x})\hat{l}_2(\vec{x})\frac{dm_3(\vec{x})}{||\vec{x}||^2}
\end{equation*}
and
\begin{equation*}
[l_2,l_2]_0=\int_{R^3}\hat{l}_1(\vec{x})\hat{l}_2(\vec{x})\frac{dm_3(\vec{x})}{1+4\pi^2||\vec{x}||^2}.
\end{equation*}
\end{theorem}
\begin{theorem}
	The space $D^*$ with the scalar products $[l_1,l_2]$, or $[l_1,l_2]_0$ is a Hilbert space.
\end{theorem}
\begin{proof} It is easy to see, that the bilinear form $[l_1,l_2]$, or $[l_1,l_2]_0$, satisfies to all properties for scalar product. Only the following one is nontrivial: from the equality $[l,l]=0$ 
it follows $l=0$. To prove the last property, let us note that 
\begin{equation*}
0=[l,l]=\int_{R^3}\left|\hat{l}(\vec{x})\right|^2\frac{dm_3(\vec{x})}{||\vec{x}||^2}
\end{equation*}
So, almost everywhere $\hat{l}(\vec{x})=0$ on $R^3$. Hence $l=0$.
\end{proof}
\begin{theorem}
	If $\vec{x}\notin supp(l)$, then $U^l,\,\,U_0^l\,\in D$ can be defined 
by the following formula
\begin{equation*}
U^l(\vec{x})=l\left(\frac{1}{||\vec{x}-\vec{y}||}\right),
\end{equation*}
and
\begin{equation*}
U_0^l(\vec{x})=l\left(\frac{\exp\{-||\vec{x}-\vec{y}||\}}{||\vec{x}-\vec{y}||}\right).
\end{equation*}
\end{theorem}
\begin{theorem}
Let us note that the potential functions $U^l,\,\,U_0^l\,\in D$, 
on $R^3\setminus supp(l)$ satisfy the conditions
\begin{equation*}
-\Delta U^l=0,\quad -\Delta U_0^l+U_0^l=0.
\end{equation*}
\end{theorem}
\begin{theorem}
Let a generalized function $l\in D^*$ have compact support. Then, in the 
sense of generalized functions we have
\begin{equation*}
-\Delta U^l=4\pi l,\quad -\Delta U_0^l+U_0^l=4\pi l.
\end{equation*}
\end{theorem}
\section{Contraction of distributions}

	Let $E\subset R^3$ be a nonempty and bounded subset of positive measure.  Note, that the characteristic function
\begin{equation*}
\chi_E(\vec{x})=1,\quad \vec{x}\in E,\quad \chi_E(\vec{x})=0,\quad \vec{x}\notin E
\end{equation*}
does not belong $D$. Therefore $l(\chi_E)$ has no meaning and so, we can not define the total 
charge concentrated on $E$. Nevertheless, it is very important to do that. 
In this section we discuss this problem.

\begin{definition}
Let $E\subset R^3$ be a nonempty and bounded subset. We say that a function   
$\varphi(\vec{x})\in D$ belongs to  $J(E)$ if for an arbitrary generalized function $l\in D^*$ 
satisfying the condition $supp(l)\subseteq E$ we have $l(\varphi)=0$. 
\end{definition}

\begin{definition}
Let $E\subset R^3$ be a nonempty and bounded subset. We say that a function   
$\varphi(\vec{x})\in D$ belongs to $I(E)$ if for an arbitrary generalized function $l\in D^*$
 satisfying the condition $supp(l)\subset \dot{E}$ we have $l(\varphi)=0$. 
\end{definition}

	Let us note, that for each nonempty and bounded subset $E\subset R^3$, the 
subspaces $I(E),\,\,J(E)$ are close in $D$ and $J(E)\subseteq I(E)$. It is possible that for some $E$ we have $J(E)\neq I(E)$.

\begin{definition}
Let $E\subset R^3$ be a nonempty and bounded subset. We denote by 
\begin{equation*}
P_{I(E)}:D\rightarrow I(E). 
\end{equation*}
the orthogonal projection on the subspace $I(E)$.
\end{definition}

\begin{definition}
Let $E\subset R^3$ be a nonempty and bounded subset. We denote by 
\begin{equation*}
P_{J(E)}:D\rightarrow J(E) 
\end{equation*}
the orthogonal projection on the subspace $J(E)$.
\end{definition}

\begin{definition}
Let $E\subset R^3$  be a nonempty and bounded subset and $l\in D^*$. 
We denote by 

\begin{equation*}
l_{I(E)}(\varphi)=l\left(\varphi-P_{I(E)}\varphi\right)
\end{equation*}
This generalized function is $I$ - contraction of $l$ on $E$.
\end{definition}

\begin{definition}
Let $E\subset R^3$  be a nonempty and bounded subset and $l\in D^*$. 
We denote by 
\begin{equation*}
l_{J(E)}(\varphi)=l\left(\varphi-P_{J(E)}\varphi\right)
\end{equation*}
This generalized function is $J$ - contraction of $l$ on $E$.
\end{definition}
	These constructions allow us to define total charge on $E$ in two manners $l_{I(E)}(1),\,\,l_{J(E)}(1)$.
Those numbers may be different.
\begin{theorem}
Let $l\in D^*$ and $E$ be a compact subset with smooth boundary 
and $supp(l)\cap \partial E=\emptyset$. Then
\begin{equation*}
l_{J(E)}(1)=-\frac{1}{4\pi}\int_{\partial E}\frac{\partial U^l(\vec{x})}{\partial \vec{n}}dm_2(\vec{x}).
\end{equation*}
\end{theorem}
\begin{proof}
By Green's formula and theorem 25 we have
\begin{equation*}
-4\pi l_{J(E)}(1)=\int_{E}\Delta U^l(\vec{x})dm_3(\vec{x})=
\int_{\partial E}\frac{\partial U^l(\vec{x})}{\partial \vec{n}}dm_2(\vec{x}).
\end{equation*}
\end{proof}
	In electrostatics, the last formula is known as Gauss theorem.
\begin{theorem}
Let $l\in D^*$ and $E$ be a compact subset with smooth boundary 
and $supp(l)\cap \partial E=\emptyset$. Then
\begin{equation*}
l_{J(E)}(1)=\frac{1}{4\pi}\int_{E}U_0^l(\vec{x})dm_3(\vec{x})-\frac{1}{4\pi}\int_{\partial E} \frac{\partial U_0^l(\vec{x})}{\partial \vec{n}}dm_2(\vec{x}).
\end{equation*}
\end{theorem}
\begin{proof}
By Green's formula and theorem 26 we have
\begin{equation*}
4\pi l_{J(E)}(1)=\int_{E}\left(U_0^l(\vec{x})-\Delta U_0^l(\vec{x})\right)dm_3(\vec{x})=
\end{equation*}
\begin{equation*}
=\int_{E}U_0^l(\vec{x})dm_3(\vec{x})-
\int_{\partial E}\frac{\partial U_0^l(\vec{x})}{\partial \vec{n}}dm_2(\vec{x}).
\end{equation*}
\end{proof}
\begin{definition}
Let $E\subset R^3$ be a compact subset and it contains a finite number of 
disjoint connected components, i. e.
\begin{equation*}
\bigcup_{k=1}^nE_k\subset E.
\end{equation*}
Let we have real numbers $q_k,\,\, k=1,2,\dots,n$. Denote by $Ch_l(E_k,q_k)$ the 
subset of generalized functions $l\in D^*$, for which
\begin{equation*}
l_{I(E_k)}(1)=q_k,\quad k=1,\dots,n.
\end{equation*}
\end{definition}
\begin{definition}
Let $E\subset R^3$ be a compact subset and it contains a finite number of 
disjoint connected components, i. e.
\begin{equation*}
\bigcup_{k=1}^nE_k\subset E.
\end{equation*}
Let we have real numbers $q_k,\,\, k=1,2,\dots,n$. Denote by $Ch_J(E_k,q_k)$ the 
subset of generalized functions $l\in D^*$, for which
\begin{equation*}
l_{J(E_k)}(1)=q_k,\quad k=1,\dots,n.
\end{equation*}
\end{definition}
\begin{theorem}
The subsets $Ch_l(E_k,q_k),\,\, Ch_J(E_k,q_k)$ are convex and close in $D^*$.
\end{theorem}
\begin{proof} It is sufficient to note that for any function $f(\vec{x})\in D$ and for 
any number $q$, the subset
\begin{equation*}
\left\{l;\,\,l(f)=q,\,\, l \in D^*\right\}
\end{equation*}
is convex and is closed subset in $D^*$. The intersection of an arbitrary number of 
such subsets preserves the above mentioned two properties.
\end{proof}
\section{Equilibrium state in the finite energy model}

	In this section we prove the existence of equilibrium distribution and its 
uniqueness.  
 \begin{definition}
Let $E\subset R^3$ be a compact subset which contain a finite number of 
disjoint connected components, i. e.
\begin{equation*}
\bigcup_{k=1}^nE_k\subset E.
\end{equation*}
and for a generalized function   
\begin{equation*}
\hat{l}_I\in Ch_l(E_k,q_k)
\end{equation*}
for which
\begin{equation*}
W(\hat{l}_I)=\inf\left\{W(l);\,\,l\in Ch_I(E_k,q_k)\right\},
\end{equation*}
then $\hat{l}_I$ is called the $I$ - equilibrium distribution.
\end{definition} 
\begin{definition}
Let $E\subset R^3$ be a compact subset which contain a finite number of 
disjoint connected components, i. e.
\begin{equation*}
\bigcup_{k=1}^nE_k\subset E.
\end{equation*}
and for a generalized function   
\begin{equation*}
\hat{l}_J\in Ch_J(E_k,q_k)
\end{equation*}
for which
\begin{equation*}
W(\hat{l}_J)=\inf\left\{W(l);\,\,l\in Ch_J(E_k,q_k)\right\},
\end{equation*}
then $\hat{l}_I$ is called the $J$ - equilibrium distribution.
\end{definition} 
\begin{definition}
The equilibrium distributions $\hat{l}_I,\,\, \hat{l}_J$ are 
said to be stable if for arbitrary 
\begin{equation*}
\hat{l}_I\neq l_1\in Ch_I(E_k,q_k),\quad \hat{l}_J\neq l_2\in Ch_J(E_k,q_k)
\end{equation*}
then we have strict inequalities
\begin{equation*}
W(\hat{l}_I)< W(l_1),\quad W(\hat{l}_J)<W(l_2).
\end{equation*}
\end{definition} 
\begin{theorem}
Let $E\subset R^3$ be a compact subset and it contains a finite 
number of disjoint connected components, i. e. 
\begin{equation*}
\bigcup_{k=1}^nE_k\subset E.
\end{equation*}
Then for any real numbers $q_k,\,\, k=1,\dots,n$, there are unique 
equilibrium distributions $\hat{l}_I,\,\, \hat{l}_J$. 
\end{theorem}
\begin{proof}
It sufficient to note, that the subsets $Ch_l(E_k,q_k),\,\, Ch_J(E_k,q_k)$ 
are convex and close. Let $l_1\in Ch_l(E_k,q_k),\,\, l_2\in Ch_J(E_k,q_k)$  be the 
distributions, where the minimum reach. Since they are unique, so for each 
\begin{equation*}
\hat{l}_I\neq l_1\in Ch_I(E_k,q_k),\quad \hat{l}_J\neq l_2\in Ch_J(E_k,q_k)
\end{equation*}
the strict inequalities 
\begin{equation*}
W(\hat{l}_I)< W(l_1),\quad W(\hat{l}_J)<W(l_2).
\end{equation*}
hold. 
\end{proof}
	The following question naturally arises: is it possible to prove 
that the equilibrium distributions are finite measures? The following example 
gives negative answer to this question. That is way we need to consider 
charges distributions as generalized functions.
\vskip0.5cm
{\bf Example 2.} 
Let a conductor $E$ be of the subset 
\begin{equation*}
E=\partial B(\vec{0},1)\cup B(\vec{0},r_1)\cup 
\bigcup_{n=1}^{\infty}\left(B(\vec{0},r_{2n+1})\setminus B(\vec{0},r_{2n})\right) 
\end{equation*} 
where $0<r_1<r_2<\dots<r_n<\dots <1$.

Put on $B(\vec{0},r_1)$ a charge equal $q$.
	In the equilibrium state on each sphere 
\begin{equation*}
\partial B(\vec{0},r_{n}),\quad n=2,3,\dots, 
\end{equation*}
inducts charges equal $q_n,\,\, n=1,\dots$.

	Let us note that the total charge placed on the closer of 
\begin{equation*}
B(\vec{0},r_{2n+1})\setminus B(\vec{0},r_{2n})
\end{equation*}
equals zero, i. e. $q_{2n+1}+q_{2n}=0,\,\, n=2,\dots$

	By the symmetry the charges are uniformly distributed on each sphere. Denote
\begin{equation*}
\mu_n(F)=\frac{m_2(F\cap \partial B(\vec{0},r_n)}{4\pi r_n^2},\quad n=1,2,\dots
\end{equation*}
The corresponding potential function of this measure equals
\begin{equation*}
U^{\mu_n}(\vec{x})=\frac{1}{r_n},\quad ||\vec{x}||\leq r_n
\end{equation*}
and
\begin{equation*}
U^{\mu_n}(\vec{x})=\frac{1}{||\vec{x}||},\quad r_n<||\vec{x}||.
\end{equation*}

	The potential function of the equilibrium distribution of all system 
permits the following representation
\begin{equation*}
U(\vec{x})=\sum_{n=1}^{\infty}q_nU^{\mu_n}(\vec{x}).
\end{equation*}

	Since the potential function   is constant on each component   . 
So, we have 
\begin{equation*}
U(\vec{x})=\sum_{k=1}^{\infty}q_{2k+1}U^{\mu_{2k+1}}(\vec{x})+
\sum_{k=1}^{\infty}q_{2k}U^{\mu_{2k}}(\vec{x})=
\end{equation*}
\begin{equation*}
=\sum_{k=1}^{n-1}\frac{q_{2k}}{||\vec{x}||}+
\sum_{k=n}^{\infty}\left(\frac{q_{2k+1}}{r_{2k+1}}+\frac{q_{2k}}{r_{2k}}\right),
\quad \vec{x}\in B(\vec{0},r_{2n+1})\setminus B(\vec{0},r_{2n}).
\end{equation*}
The above - mentioned conditions can be valid only if 
\begin{equation*}
\sum_{k=1}^{n-1}q_{2k}=0,\quad n=2,\dots
\end{equation*}
Finally we get 
\begin{equation*}
q_n
=(-1)^{n-1}q_1,\quad n=1,2,\dots
\end{equation*}

	From this result we conclude that the equilibrium distribution for 
a given conductor can not be a finite measure.

	Let us note, that if by thin wire we will connect the inside surfaces 
and by another thin wire we will connect the outside surfaces, then on the 
ends of those wires will arises an arbitrary big potential drop. 
\begin{theorem}
Let the conductor $E$ contain a finite number of connected 
components, i. e. 
\begin{equation*}
\bigcup_{k=1}^nE_k\subset E.
\end{equation*}
The potential function of the equilibrium distribution $l_0\in Ch_I(E_k,q_k)$ 
is constant in the interior points of each component $E_k,\,\, k=1,\dots,n$. 
\end{theorem}
\begin{proof} Let $l_0\in Ch_I(E_k,q_k)$ be an equilibrium distribution. 
Suppose that the corresponding potential function is not constant, i.e. there 
is an index $k$ and there are disjoint balls 
\begin{equation*}
B(\vec{x}_1,r)\subset E_k,\quad B(\vec{x}_2,r)\subset E_k
\end{equation*}
such that the inequality 
\begin{equation*}
\int_{B(\vec{x}_1,r)}U^{l_0}(\vec{x})dm_3(\vec{x})< \int_{B(\vec{x}_2,r)}U^{l_0}(\vec{x})dm_3(\vec{x})
\end{equation*}
holds.  

	Let $\varphi(\vec{x})\geq 0\in D$ be a nonzero function such 
that $supp(\varphi)\subset B(\vec{0},r)$. Let us put
\begin{equation*}
\varphi_1(\vec{x})=\varphi(\vec{x}-\vec{x}_1),\quad \varphi_2(\vec{x})=\varphi(\vec{x}-\vec{x}_2)
\end{equation*}
Let us note that $supp(\varphi_1)\subset B(\vec{x}_1,r),\,\,\,supp(\varphi_2)\subset B(\vec{x}_2,r)$.
For arbitrary number $a$ we have 
\begin{equation*}
l_0+al_{\varphi_1}-al_{\varphi_2}\in Ch(E_k,q_k)
\end{equation*}
and the inequality
\begin{equation*}
[l_0,l_{\varphi_1}]=\int_{B(\vec{x}_1,r)}U^{l_0}(\vec{x})\varphi_1(\vec{x})dm_3(\vec{x})<
\end{equation*}
\begin{equation*}
<\int_{B(\vec{x}_2,r)}U^{l_0}(\vec{x})\varphi_2(\vec{x})dm_3(\vec{x})=[l_0,l_{\varphi_2}]
\end{equation*}
holds. Since $l_0$ be an equilibrium distribution so, we have the following inequality
\begin{equation*}
W(l_0)\leq W(l_0+al_{\varphi_1}-al_{\varphi_2})
\end{equation*}
On the other hand we have 
\begin{equation*}
W(l_0+al_{\varphi_1}-al_{\varphi_2})=
\end{equation*}
\begin{equation*}
=W(l_0)+2a([l_0,l_{\varphi_1}]-[l_0,l_{\varphi_2}])
+a^2[l_{\varphi_1},l_{\varphi_1}]-
2a^2[l_{\varphi_1},l_{\varphi_2}]+a^2[l_{\varphi_2},l_{\varphi_2}].
\end{equation*}
For sufficiently small values of the parameter $0<a$ we have 
\begin{equation*}
W(l_0+al_{\varphi_1}-al_{\varphi_2})<W(l_0).
\end{equation*}
This contradiction proves theorem.
\end{proof} 
\begin{theorem} 
	For any compact subset $E$ the corresponding equilibrium 
distribution $l_0\in Ch_I(E_k,q_k)$ has the property $supp(l_0)\subseteq \partial E$.
\end{theorem}
\begin{proof} Let $l_0\in Ch_I(E_k,q_k)$ be the equilibrium distribution. Since the potential 
function $U^{l_0}$ is constant on each component $E_k$, so for each 
$\varphi\in D,\,\,\, supp(\varphi)\subset \dot{E}_k$, we have
\begin{equation*}
\Delta U^{l_0}(\varphi)=0.
\end{equation*}
By preceding theorem we have
\begin{equation*}
\Delta U^{l_0}(\varphi)=-\frac{1}{4\pi}l_0(\varphi).
\end{equation*}
Consequently, for arbitrary $\varphi\in D$ satisfying the condition $supp(\varphi)\subseteq E_k$ we have $l_0(\varphi)=0$. 
\end{proof}
	Theorem suggests that in the finite energy model, there is an 
equilibrium distribution. It is unique and it is stable.

	The second part of the Theorem explains the Cavendish's experiment. 

\section{Forces in the finite energy model}

\begin{definition}
We say that testing function $\varphi(\vec{x})$, belongs to 
the space $V$ if $\nabla \varphi(\vec{x})\in D$, it has compact support and
\begin{equation*}
||\nabla \varphi(\vec{x})||<\infty.
\end{equation*}
\end{definition}
\begin{definition}
Let $l\in D^*$ have a compact support and its potential 
function $U^l(\vec{x})$ is bounded. The forces distribution, for a 
generalized function $l$, we define as a new generalized function $\vec{F}_l$ 
acts on testing function $\varphi(\vec{x})\in V$ as follows 
\begin{equation*}
\vec{F}_l(\varphi)=l\left(U^l\nabla \varphi\right),\quad \varphi\in V.
\end{equation*}
\end{definition}

	Note that we can determine the forces, in the finite energy model, only if 
the potential function $U^l$ is bounded. 
\begin{definition}
Let $l\in D^*$ be a generalized function and $supp(l)\subset E$. 
Let $\vec{x}_0\in E$ and $\vec{n}$ be a unit vector. We say, that at the 
point $\vec{x}_0$ vector field $\vec{F}_l$ has a 
nontrivial component depth ward $\vec{n}$, if there is a constant $0<a$ such that 
for an arbitrary $0<r$ there is a function $\varphi_0(\vec{x})\in V$ with
\begin{equation*}
supp(\varphi_0)\subset B(\vec{x}_0,r),\quad ||\varphi||+\sup_{\vec{x}}|\varphi(\vec{x})|\leq 1
\end{equation*}
and satisfying the condition
\begin{equation*}
l\left(U_l\frac{\partial \varphi_0}{\partial \vec{n}}\right)>a.
\end{equation*}
\end{definition}
\begin{theorem}
Let $E$ be compact subset. Then for equilibrium distribution $l\in D^*$ 
the forces $\vec{F}_l$ at the point $\vec{x}_0\in E$ can not have a nontrivial 
component depth ward an arbitrary for $E$ inner direction.
\end{theorem}
\begin{proof}  Let $l\in D^*$ be an equilibrium distribution. Let $\vec{n}$ be a unit
inner vector for $E$ at the point $\vec{x}_0$, i.e. for each positive number $0<\epsilon$ 
there are $0<r<t<\epsilon$ such that
\begin{equation*}
B(\vec{x}_0+t\vec{n},r)\subset E
\end{equation*}

	Let us assume, that at the point $\vec{x}_0$ the force $\vec{F}_l$, 
has a nontrivial component depth ward $\vec{n}$. This follows that there 
is a constant $a>0$ such that for an 
arbitrary $r>0$ there is a function $\varphi_0(\vec{x})\in V$ with
\begin{equation*}
supp(\varphi_0)\subset B(\vec{x}_0,r),\quad ||\varphi_0||+\sup_{\vec{x}}|\varphi_0(\vec{x})|\leq 1
\end{equation*}
and satisfying the condition
\begin{equation*}
l\left(U_l\frac{\partial \varphi_0}{\partial \vec{n}}\right)>a.
\end{equation*}

	The generalized function $l$ permits the following representation 
\begin{equation*}
l(\psi)=l(\psi\varphi)+l(\psi(1-\varphi)).
\end{equation*}

	Let us introduce new generalized function acting on each test 
function   as follows
\begin{equation*}
l_t(\psi)=l(\psi\varphi_t)+l(\psi(1-\varphi)).
\end{equation*}
where 
\begin{equation*}
\varphi_t(\vec{x})=\varphi(\vec{x}-t\vec{n}).
\end{equation*}
We have
\begin{equation*}
\frac{\partial \varphi_t(\vec{x})}{\partial \vec{n}}=-\frac{\partial \varphi_t(\vec{x})}{\partial t},
\end{equation*}
where the left hand side is the derivative of the function $\varphi_t(\vec{x})$ by  the 
variable $\vec{x}$ in direction $\vec{n}$.

	Since 
\begin{equation*}
B(\vec{x}_0+t\vec{n},r)\subset E
\end{equation*}
so we have $l_t\in Ch_I(E_k,q_k)$. Consequently, $W(l)<W(l_t)$ and
\begin{equation*}
U^l(\vec{x})=const,\quad \vec{x}\in \dot{E}.
\end{equation*}
We have
\begin{equation*}
l_t(\psi)=l(\psi)+l((\varphi_t-\varphi\psi)=l(\psi)+\tilde{I}(\psi).
\end{equation*}
So, taking into account the condition $U^l(\vec{x})=const,\,\,\,\vec{x}\in \dot{E}$ we have
\begin{equation*}
W(l_t)-W(l)=l_t(U^{l_t})-l(U^{l})=2l\left((\varphi_t-\varphi)U^l\right)+
2l\left((\varphi_t-\varphi)U^{\tilde{l}}\right)=
\end{equation*}
\begin{equation*}
=2\tilde{l}(U^{l_t})+\tilde{l}(U^{\tilde{l}})=2l\left((\varphi_t-\varphi)U^l\right)
+l\left((\varphi_t-\varphi)U^{\tilde l}\right)=
\end{equation*}
\begin{equation*}
=2l\left((\varphi_t-\varphi)U^l\right)
+l\left((\varphi_t-\varphi)^2U^l\right)=
-2tl\left(U^l\frac{\partial \varphi}{\partial \vec{n}}\right)+o(t)<0,\quad t\to +0.
\end{equation*}
The getting inequality contradicts our chose of $l\in D^*$ to be 
an equilibrium distribution.
\end{proof}
\begin{definition}
Let $E$ be a compact subset. We say, that a distribution $l\in D^*$ 
with $supp(l)\subset E$ is in a static state, if the force $\vec{F}$ has no inner 
direction at each point $\vec{x}\in E$.
\end{definition}
	In particularly, for each test function $\varphi$, satisfying the 
condition $supp(\varphi)\subset \dot{E}$, we have
\begin{equation*}
\vec{F}(\varphi)=0.
\end{equation*}
This follows that there is a constant number $C$ such that if $supp(\varphi)\subset \dot{E}$, then
\begin{equation*}
l(U^l\psi)=Cl_1(\psi).
\end{equation*}

	From the given bellow two examples follow that static state is not unique.
\vskip 0.5cm
	{\bf Example 3.} Let we have the distribution
\begin{equation*}
l(\varphi)=\frac{Q}{4\pi r_0^2}\int_{\partial B(\vec{0},r_0)}\varphi(\vec{x})dm_2(\vec{x}).
\end{equation*}
The total charge equals $l(1)=Q$ and the corresponding potential function is
\begin{equation*}
U^l(\vec{x})=\frac{Q}{r_0},\quad ||\vec{x}||\leq r_0,
\end{equation*}
\begin{equation*}
U^l(\vec{x})=\frac{Q}{||\vec{x}||},\quad r_0<||\vec{x}||,
\end{equation*}
This is static state.

\section{Equilibrium distribution on two balls}
 
	In this section we determine Kelvin's transform and some of its important 
and useful property, see \cite {b:J}.
\begin{definition}
Let $0<||\vec{x}-\vec{x}_0||<R$. Let 
\begin{equation*}
\vec{y}=\vec{x}_0+\frac{R^2}{||\vec{x}-\vec{x}_0||}(\vec{x}-\vec{x}_0)
\end{equation*}
be the Kelvin,s transform of the point $\vec{x}$ with respect to the sphere $\partial B(\vec{x}_0,R)$. 
\end{definition}
	Note, that in inverse transform, all points situated on the sphere $\partial B(\vec{x}_0,R)$, 
remain stationary and the points $\vec{x}_0,\,\,\vec{x}\,\,\vec{y}$ lie on a straight line.

	Furder, we have  
\begin{equation*}
||\vec{x}-\vec{x}_0||||\vec{y}-\vec{x}_0||=R^2
\end{equation*}

	It is well known the following property of Kelvin's transform.

	Let $0<||\vec{x}-\vec{x}_0||<R$, and $\vec{y}$ be the Kelvin's transform of $\vec{x}$ 
with respect to the sphere $\partial B(\vec{x}_0,R)$. Then for an arbitrary 
point $\vec{z}\in\partial B(\vec{x}_0,R)$ we have the following equality
\begin{equation*}
\frac{1}{||\vec{z}-\vec{y}||}=\frac{||\vec{x}-\vec{x}_0||}{R}\frac{1}{||\vec{z}-\vec{x}||}
\end{equation*}

	Let us consider the conductor of the form $E=B(\vec{x}_0,R)\cup B(\vec{y}_0,r)$. 
The first ball $B(\vec{x}_0,R)$ has a positive charge equals $Q$ and the second 
ball $B(\vec{y}_0,r)$ has a positive charge $r$.

	Let us denote  
\begin{equation*}
d=||\vec{y}_0-\vec{x}_0||-r-R>0
\end{equation*}
We want to find the equilibrium distribution. Denote by
\begin{equation*}
\vec{x}_1=\vec{x}_0+(\vec{y}_0-\vec{x}_0)\frac{R^2}{||\vec{y}_0-\vec{x}_0||^2}
\end{equation*}
the point symmetric to $\vec{y}_0$, with respect to the ball $B(\vec{x}_0,R)$. 

	Denote by
\begin{equation*}
\vec{y}_1=\vec{y}_0+(\vec{x}_0-\vec{y}_0)\frac{r^2}{||\vec{y}_0-\vec{x}_0||^2}
\end{equation*}

the point symmetric to $\vec{x}_0$ with respect to the ball $B(\vec{y}_0,r)$.

	Similarly, by induction we define the points $\vec{x}_n$ symmetric to $\vec{y}_{n-1}$ with 
respect to the ball $B(\vec{x}_0,R)$ 
\begin{equation*}
\vec{x}_n=\vec{x}_0+(\vec{y}_{n-1}-\vec{x}_0)\frac{R^2}{||\vec{y}_{n-1}-\vec{x}_0||^2},\quad n=1,2,\dots
\end{equation*}
and the points $\vec{y}_n$ symmetric to $\vec{x}_{n-1}$, with respect to the ball $B(\vec{y}_0,r)$, i.e.
\begin{equation*}
\vec{y}_n=\vec{y}_0+(\vec{x}_{n-1}-\vec{x}_0)\frac{r^2}{||\vec{x}_{n-1}-\vec{x}_0||^2},\quad n=1,2,\dots
\end{equation*}

	Note that
\begin{equation*}
||\vec{x}_n-\vec{x}_0||=\frac{R^2}{||\vec{y}_{n-1}-\vec{x}_0||},\quad n=1,2,\dots
\end{equation*}
and
\begin{equation*}
||\vec{y}_n-\vec{y}_0||=\frac{r^2}{||\vec{x}_{n-1}-\vec{y}_0||},\quad n=1,2,\dots
\end{equation*}

	All points $\vec{x}_0,\,\,\vec{x}_1\dots$ lie inside the ball $B(\vec{x}_0,R)$ and all 
points $\vec{y}_0,\,\,\vec{y}_1\dots$ lie in the ball $B(\vec{y}_0,r)$.
 
	Since for any $n=0,1,\dots$ the points $\vec{x}_{n+1},\,\,\,\vec{y}_{n}$  are symmetric 
with respect to the ball $B(\vec{x}_0,R)$, so
\begin{equation*}
\frac{1}{||\vec{x}-\vec{y}_n||}=\frac{||\vec{x}_{n+1}-\vec{x}_0||}{R}\frac{1}{||\vec{x}_{n+1}-\vec{x}||},\quad ||\vec{x}-\vec{x}_0||=R.
\end{equation*}
Since for any $n=0,1,\dots$ the points $\vec{x}_{n},\,\,\,\vec{y}_{n+1}$ are symmetric 
with respect to of the ball $B(\vec{y}_0,r)$, so
\begin{equation*}
\frac{1}{||\vec{x}-\vec{x}_n||}=\frac{||\vec{y}_{n+1}-\vec{y}_0||}{R}\frac{1}{||\vec{y}_{n+1}-\vec{x}||},\quad ||\vec{x}-\vec{y}_0||=r.
\end{equation*}
In figure 4 we show only four points.

	Define the potential function outside of conductor $\Omega=R^3\setminus E$ in the 
following form
\begin{equation*}
U(\vec{x})=\frac{C}{||\vec{x}-\vec{x}_0||}+\sum_{n=0}^{\infty}C_n
\left(\frac{1}{||\vec{x}-\vec{y}_n||}-\frac{||\vec{x}_{n+1}-\vec{x}_0||}{R}\frac{1}{||\vec{x}-\vec{x}_{n+1}||}\right)
\end{equation*}
Note that 
\begin{equation*}
U(\vec{x})=\frac{C}{R},\quad ||\vec{x}-\vec{x}_0||=R.
\end{equation*}
The total charge, placed inside the ball $B(\vec{y}_0,r)$ equals
\begin{equation*}
q=\sum_{n=0}^{\infty}C_n.
\end{equation*}

	We assume, that the same potential function permits the following 
representation, too 
\begin{equation*}
U(\vec{x})=\frac{D}{||\vec{x}-\vec{y}_0||}+\sum_{n=0}^{\infty}D_n
\left(\frac{1}{||\vec{x}-\vec{x}_n||}-\frac{||\vec{y}_{n+1}-\vec{y}_0||}{r}\frac{1}{||\vec{x}-\vec{y}_{n+1}||}\right)
\end{equation*}
Note that 
\begin{equation*}
U(\vec{x})=\frac{D}{r},\quad ||\vec{x}-\vec{y}_0||=r.
\end{equation*}
Inside the ball $B(\vec{x}_0,R)$ we have the charge 
\begin{equation*}
Q=\sum_{n=0}^{\infty}D_n.
\end{equation*}
The same potential function $U(\vec{x}$ permits both of the above mentioned 
representations, if
\begin{equation*}
\frac{C-D_0}{||\vec{x}-\vec{x}_0||}-\sum_{n=1}^{\infty} \frac{1}{||\vec{x}-\vec{x}_n||}
\left(D_n+C_{n-1}\frac{||\vec{x}_0-\vec{x}_n||}{R}\right)=
\end{equation*}
\begin{equation*}
=\frac{D-C_0}{||\vec{x}-\vec{y}_0||}-\sum_{n=1}^{\infty} \frac{1}{||\vec{x}-\vec{y}_n||}
\left(C_n+D_{n-1}\frac{||\vec{y}_0-\vec{y}_n||}{r}\right)
\end{equation*}
These equalities are valid if $D_0=C,\,\,\, C_0=D$ and
\begin{equation*}
D_n=-C_{n-1}\frac{||\vec{x}_n-\vec{x}_0||}{R},\quad C_n=-D_{n-1}\frac{||\vec{y}_n-\vec{y}_0||}{R},\quad n=1,2,\dots     
\end{equation*}
Denote
\begin{equation*}
D_n=C\hat{D}_n,\quad C_n=D\hat{C}_n,
\end{equation*}
where $\hat{D}_n,\quad \hat{C}_n$ do not depend on the parameters $D,\,\,C$. Consequently, we get the 
following equations
\begin{equation*}
Q=\frac{C}{2}\left(\sum_{n=0}^{\infty}(1+(-1)^n)\hat{D}_n\right)+\frac{D}{2}
\left(\sum_{n=0}^{\infty}(1-(-1)^n)\hat{D}_n\right)=A_{11}C+A_{12}D
\end{equation*}
\begin{equation*}
q=\frac{C}{2}\left(\sum_{n=0}^{\infty}(1-(-1)^n)\hat{C}_n\right)+\frac{D}{2}
\left(\sum_{n=0}^{\infty}(1+(-1)^n)\hat{C}_n\right)=A_{21}C+A_{22}D
\end{equation*}
So, we have
\begin{equation*}
C=\frac{A_{22}Q-A_{12}q}{A_{11}A_{22}-A_{12}A_{21}},\quad D=\frac{A_{11}q-A_{21}Q}{A_{11}A_{22}-A_{12}A_{21}}
\end{equation*}

\section{Jagged effect}
  
	Consider the case $R=r$ and $Q=q$. The potential function on the segment
\begin{equation*}
I=\left\{t\vec{x}_0+(1-t)\vec{y}_0;\,\,\,1-\frac{r}{||\vec{x}_0-\vec{y}_0||}\leq t\leq \frac{r}{||\vec{x}_0-\vec{y}_0||}\right\}
\end{equation*}
is outside of balls. The potential function on   is represented in the 
following picture
 
\begin{figure}[tbp]
\begin{center}
\includegraphics[scale = 0.8]{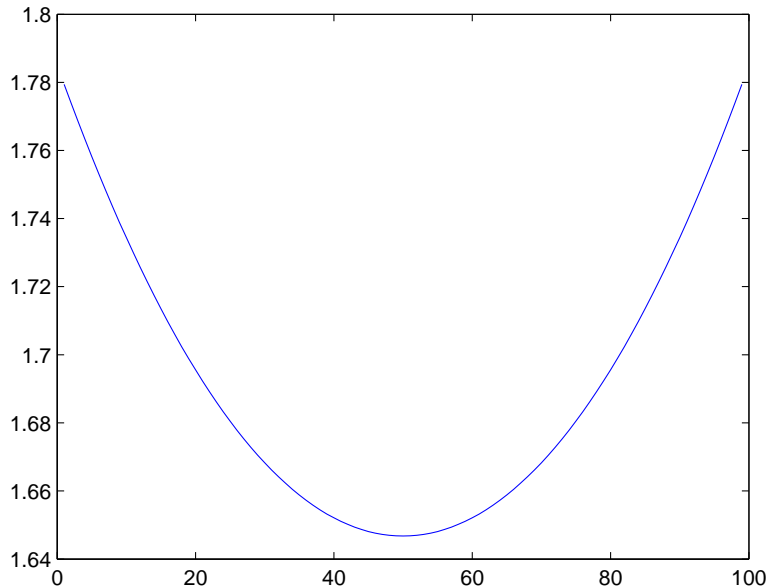}
\end{center}
\caption{The typical form of the potential function on the line between balls.}
\end{figure}

In the figure 2. the oscillation  
\begin{equation*}
E(d)=\max_{\vec{x},\vec{y}\in I}|U(\vec{x})-U(\vec{y})|
\end{equation*}
as a function of the distance $d$ between ball centers  
\begin{equation*}
||\vec{x}_0-\vec{y}_0||=d+2r>2r,
\end{equation*}
is presented.

\begin{figure}[tbp]
\begin{center}
\includegraphics[scale = 0.8]{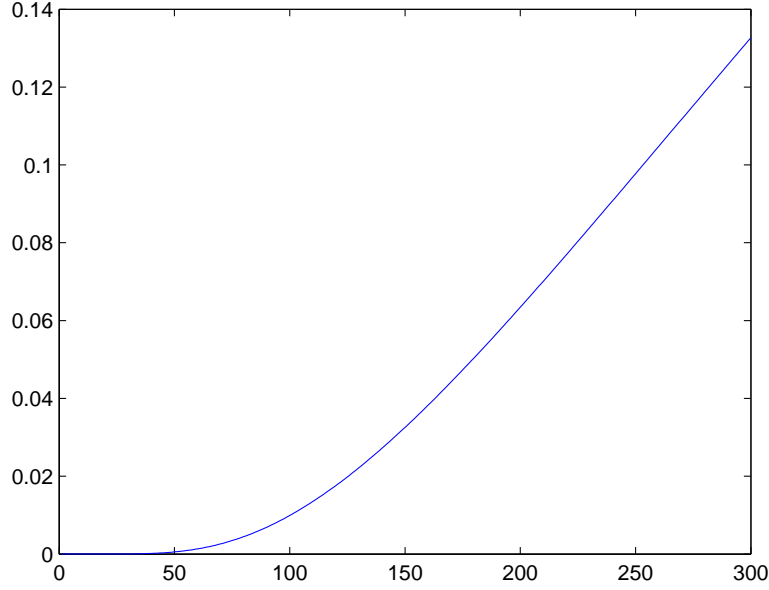}
\end{center}
\caption{The oscillation of potential function for two balls as a function 
of its distance.}
\end{figure}

	We see that on the line which connect the nodes centers the potential 
function has small oscillation. Let us consider two connected 
subsets $E_1,\,\,\,E_2$ which are situated inside the disjoint balls
\begin{equation*}
E_1\subset B(\vec{x}_1,r),\quad E_2\subset B(\vec{x}_2,r).
\end{equation*}
We put the same charges on these subsets. In this section we prove that 
the shape of boundaries  $\partial E_1,\,\,\,\partial E_2$ play an essential role. 

\begin{figure}[tbp]
\begin{center}
\includegraphics[scale = 0.8]{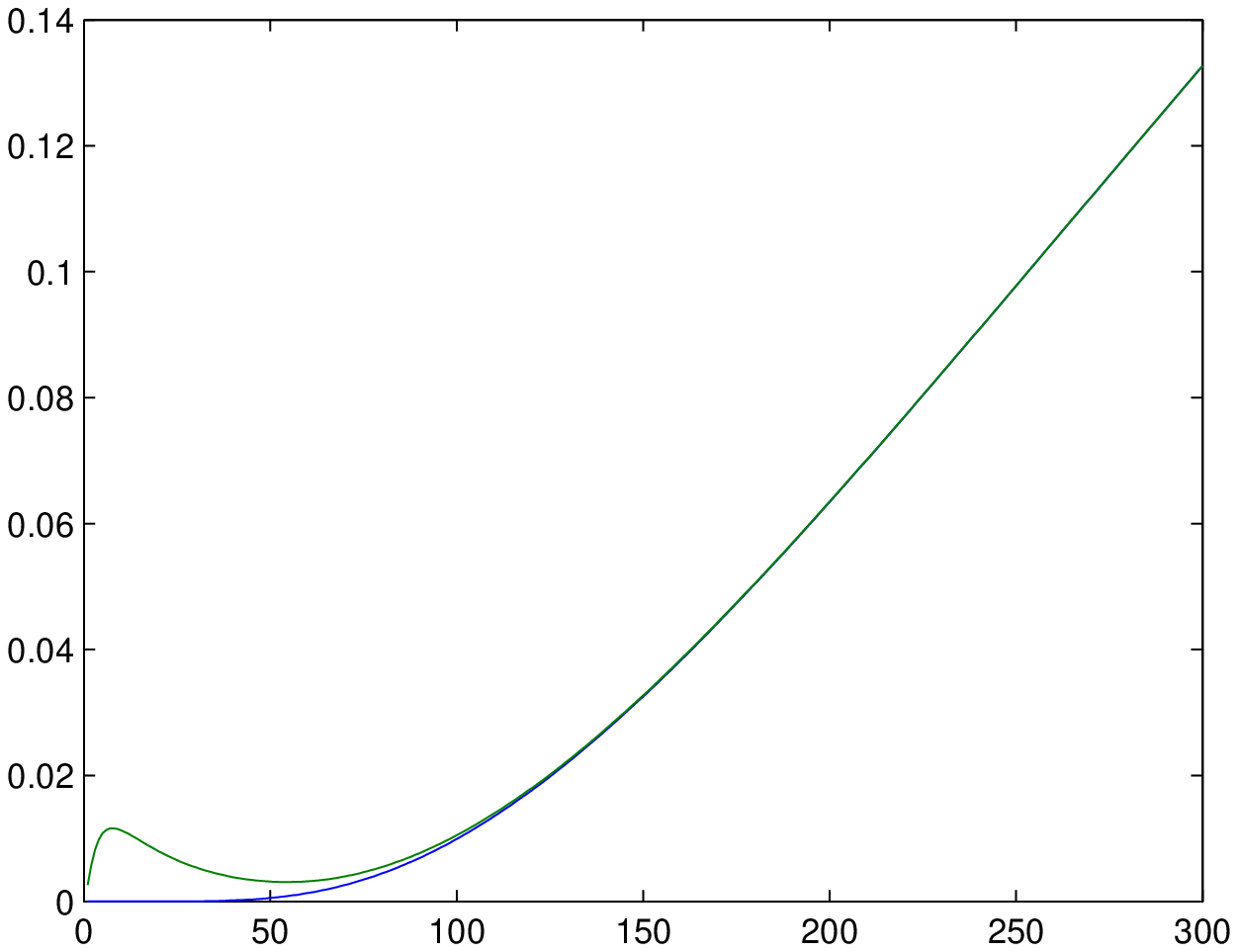}
\end{center}
\caption{The possible values of oscillation the potential function for 
two component conductors for different shapes.}
\end{figure}

\begin{figure}[tbp]
\begin{center}
\includegraphics[scale = 0.4]{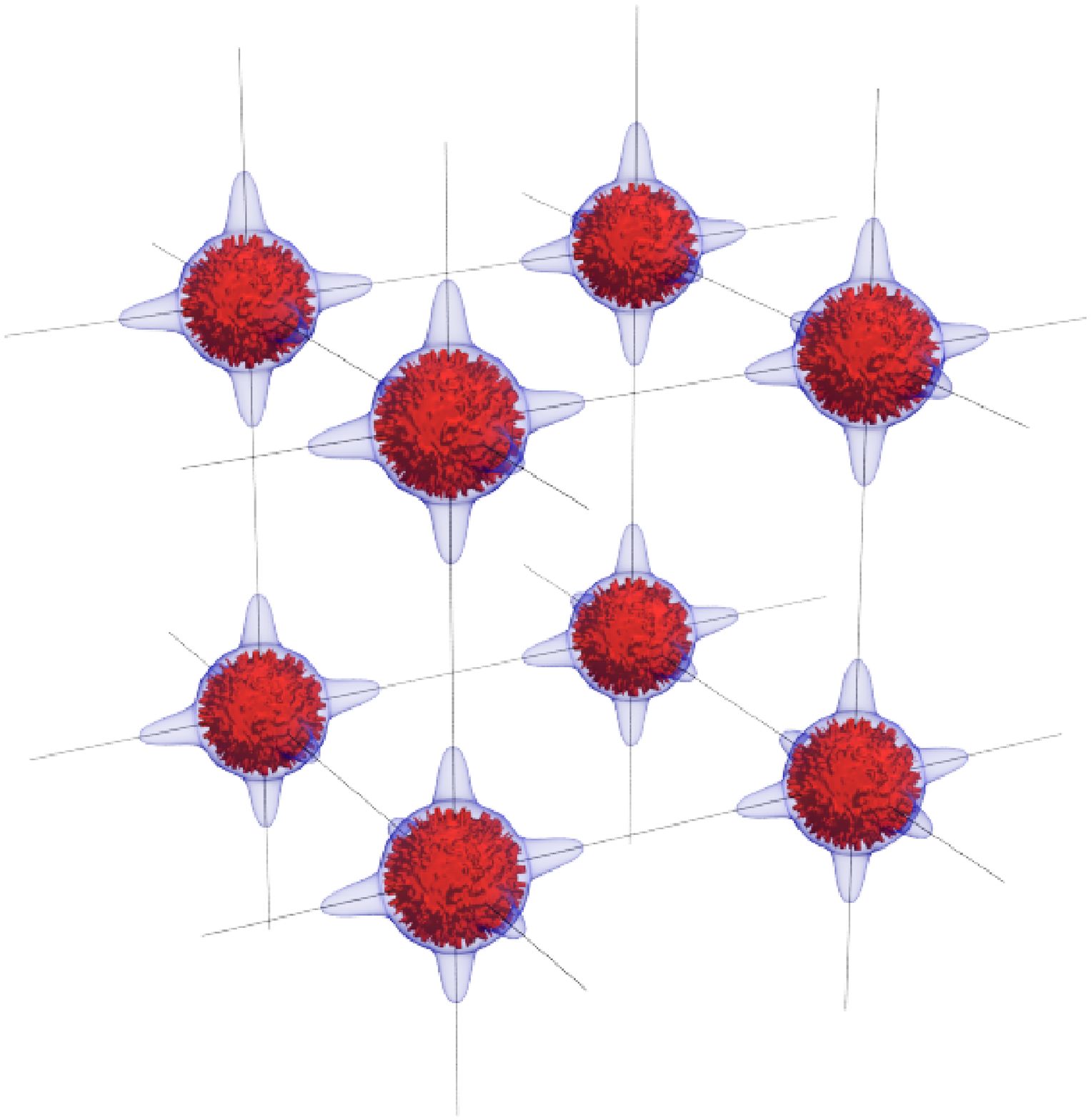}
\end{center}
\caption{The nodes in height temperature.}
\end{figure}

In height temperature, the nodes lose the ideal form and as a result 
potential function has big oscillation. So, the equipotent surface of 
potential function can not go far  from nodes and  connect the neighborhood 
placed nodes. This result we name jagged effect. In the next figure we 
present a typical situation.

	Let the same positive charges $q$ be placed at the points
\begin{equation*}
\vec{x}_{k,j}=
(r\sin\varphi_k\cos\theta_{k,j},\,\,r\sin\varphi_k\sin\theta_{k,j},\,\,r\cos \varphi_k)\in \partial B(\vec{0},r),
\end{equation*}
where 
\begin{equation*}
\varphi_k=\frac{\pi}{2}+\frac{\pi k}{2n},\quad k=-n,\dots,-1,0,1,\dots,n,
\end{equation*}
and
\begin{equation*}
\theta_{k,j}=\frac{2\pi j}{n\sin \varphi_k},\quad j=1,\dots,[n\sin\varphi_k].
\end{equation*}
We have $N=n^2$ separated points.  Moreover, we have 
\begin{equation*}
||\vec{x}_{k,i}-\vec{x}_{k,j}||^2=
r^2\sin^2\varphi_k\left[\left(\cos\theta_{k,i}-\cos\theta_{k,j}\right)^2+\left(\sin\theta_{k,i}-\sin\theta_{k,j}\right)^2\right]=
\end{equation*}
\begin{equation*}
=4r^2(\sin^2\varphi_k)\sin^2\left(\frac{2\pi}{n\sin\varphi_k}\right)\approx \frac{8\pi r^2}{n^2}\approx\frac{8\pi r^2}{N}.
\end{equation*}
and
\begin{equation*}
||\vec{x}_{k,j}-\vec{x}_{m,j}||^2=
r^2\left(\cos\varphi_k-\cos\varphi_m\right)^2+r^2\left(\sin\varphi_k-\sin\varphi_m\right)^2\geq \frac{r^2}{n^2}\approx\frac{r^2}{N}
\end{equation*}
Consequently, for different indexes we have
\begin{equation*}
||\vec{x}_{k,j}-\vec{x}_{m,j}||\geq \frac{r}{\sqrt N}
\end{equation*}

	It is easy to verify that we have the inclusion 
\begin{equation*}
\partial B(\vec{0},r)\subset \bigcup_{k,j}B\left(\vec{x}_{k,j},\frac{4r}{\sqrt N}\right)
\end{equation*}

	Let us put the same positive charges $q$ at the points 
\begin{equation*}
\vec{d}+\vec{x}_{k,j}\in \partial B(\vec{d},r)
\end{equation*}
where $||\vec{d}||>r$. Denote the set
\begin{equation*}
E(N)=\left\{\vec{x};\quad \sqrt N \leq \sum_{k,j}\frac{q}{||\vec{x}_{k,j}-\vec{x}||}+
\sum_{k,j}\frac{q}{||\vec{x}_{k,j}-\vec{d}-\vec{x}||}\right\}.
\end{equation*}

	The subset $E(N)$ consists of three connected components. One of those 
components is unbounded. Denote by $G(N)$ that unbounded component. We have 
the following representation
\begin{equation*}
R^3\setminus G(N)=E\cup F,
\end{equation*}
where $E,\,\, F$ are disjoint connected subsets. For concreteness, suppose 
that $\vec{0}\in E$ and $\vec{d}\in F$. 

Let 
\begin{equation*}
U(\vec{x})=\sqrt N,\quad \vec{x}\notin G(N)
\end{equation*}
and
\begin{equation*}
U(\vec{x})=\sum_{k,j}\frac{q}{||\vec{x}_{k,j}-\vec{x}||}+
\sum_{k,j}\frac{q}{||\vec{x}_{k,j}-\vec{d}-\vec{x}||},\quad \vec{x}\in G(N).
\end{equation*}

	Note that $U(\vec{x})$ is the potential function of the equilibrium distribution 
of the charge $qN$ placed on $E$ and the same charge placed on $F$. 

	It is easy to see that the pieces $E$ and $F$ repel each other because 
they contain only positive charges.

	For any $\epsilon>0$ you can select a number $N$ so large that 
\begin{equation*}
B(\vec{0},r-\epsilon)\subset E_1\subset B(\vec{0},r+\epsilon),\quad 
B(\vec{d},r-\epsilon)\subset E_2\subset B(\vec{d},r+\epsilon).
\end{equation*}

\section{Basic experimental facts on superconductivity}

	Superconductivity is one of the most fascinating chapters of modern physics. 

	During the past century, enormous number of experimental results where gathered. 
Below we present only those, which we can explain in frame of 
suggested in this paper new model.

1.	The existence of critical temperature. 

	In 1911 K. Ones discovered that at a critical 
low temperature the resistance of Hg suddenly falls to zero. 

	More accurate experiments give the value $10^{-24}$ Ohm 
for resistance of Hg at the critical temperature 4T, and the value $10^{-9}$ Ohm 
at the temperature 4.2T.
 
	Now, this phenomenon is known as superconductivity.

	The property of superconductivity was observed for the 
following metals: Al,  Cd,  Ga,  Hf,  Hg,  In,  Ir,  La,  Mo,  Mb,  Os,  
Pa,  Pb,  Re,  Ru,  Sn,  Ta,  Tc,  Th,  Ti,  Tl,  U,  V,  W,  Zn,  Zr. 
Later one discovers this effect for some alloys and ceramic materials, too.
 
2.	A small increasing of the resistance, before the critical temperature. 

	For some materials, near the critical temperature, the resistance suddenly 
increases a little and reaching some maximum value quickly drops to zero, 
see [22], p.436, see picture 5. 

\begin{figure}[tbp]
\begin{center}
\includegraphics[scale = 0.7]{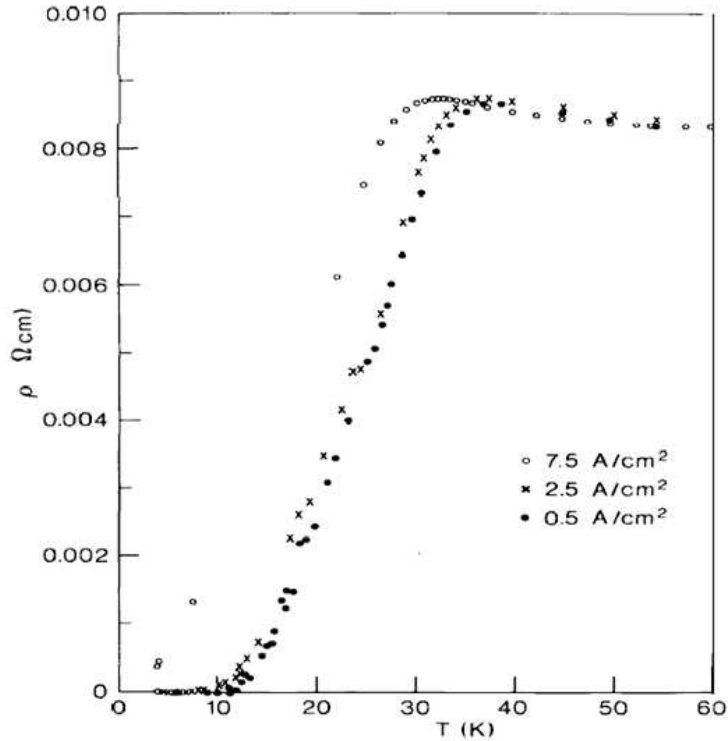}
\end{center}
\caption{Low temperature resistive of a sample recording for different current densities. From \cite {b:B}.}
\end{figure}

3.	The role of pressure. 

	For the following elements: 
Be, Cr, Ba, Si, Ge,  Se, Sb, Te, Bi the superconductivity 
property was observed only in condition of the huge pressure.
 
4.	The role of magnetic field.

	Experimentally it was verified that the 
superconductivity is destroyed in sufficiently strong magnetic field.

5. The role of current's magnitude.	

Superconductivity is destroyed also, 
when the current is greater of some critical value. This result is known 
as Silcbay effect.

\begin{figure}[tbp]
\begin{center}
\includegraphics[scale = 0.4]{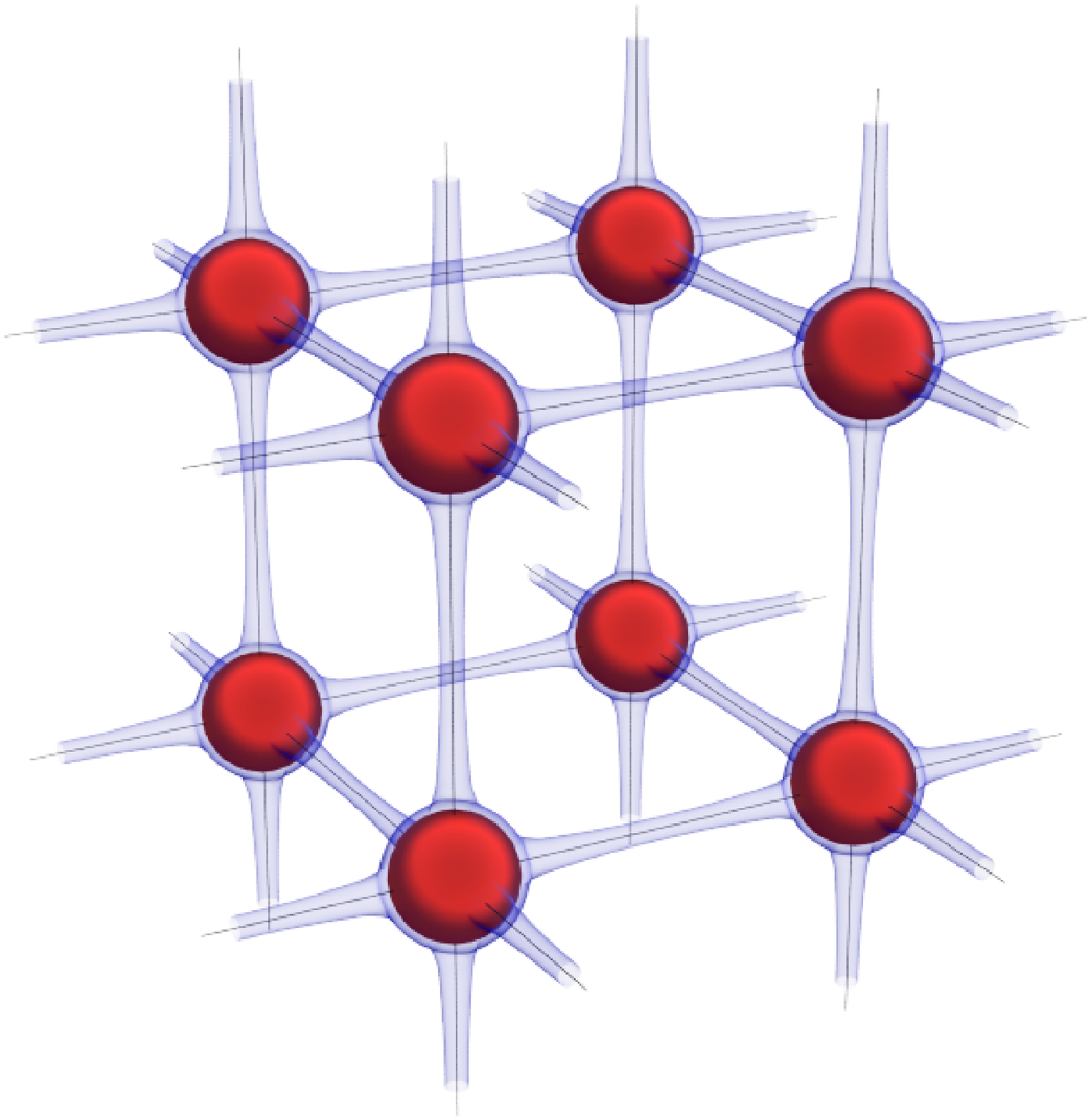}
\end{center}
\caption{The nodes at low temperature.}
\end{figure}

\section{New model for metal}

	Here we give the main postulates of the new model, see \cite {b:V}. 

	1. The metal has a crystal structure. 

In 1912 Laue talked a report in Bavarian Academy of sciences about 
interference of Rentgen's array. He announced that the experimental results show, that 
metal consists of crystal lattice. 
The positively charged ions are relatively 
immobile and form nodes with positive total charge. 
On each node there are relatively free electrons. Those electrons, are named semi - free 
electrons. They can freely move thought node, but to leave the node semi - free electron 
needs some additional energy.

	2. The existence of free electrons cloud.

In a metal there is a cloud of free electrons. An elegant 
experiment confirm the existence of free electrons was proposed by 
J. Maxwell. That experiment is not easy to realize because of the weak 
expected effect. Nevertheless, it was done after J. Maxwell (1831 - 1879) 
death, by Tolmen and Stewart in 1916. They built a coil with many turns 
and put it in rapid rotation. When the coil suddenly stops, through it 
passes a current. It was found, that through the coil move particles 
with negative charges. This experiment demonstrates the relatively 
independence of free electrons cloud and the crystal lattice.

3. We assume that the electric conductions are caused by free 
electrons motion.

	If we switch on an outside electrical field, then in the cloud of 
free - electrons a directed motion arises. That wind is interpreted 
as an electrical current.

4.  The resistance is conditioned by chaotic motion of free electrons.

	If a free electron, by outside influence, gains a directed motion component, then 
due time it will lose energy. That is caused of the chaotically moving 
free electrons medium. After some 
time period the directed component will vanishes. 
This effect is the cause of resistance. 

This remark follows that if in the given metal the number of free electrons are much more than 
semi - free electrons, then it is good conductor. Similarly, if in the given metal the free electrons 
are less than semi - free electrons, then it is bad conductor.

5. Some metals, which are not superconductors, gain that property 
condition of the huge pressure, see \cite {b:Kob}.

\section{Discussion of experimental results}
	The domains, occupied by the nodes of crystal, let us denote by
\[
E_n(t,T),\quad n=1,2,\dots, N,
\]
where $t$ is the time and $T$ is the temperature.

	We assume that there are immobile points $\vec{x}_n,\quad n=1,2,\dots, N,$ such that
\[
B(\vec{x}_n,r-\epsilon)\subseteq E_n(t,T)\subseteq B(\vec{x}_n,r+\epsilon),
\]
where $0<\epsilon$ is a small number.  We assume that at each time moment $t$ the 
semi - free electrons are in equilibrium state.
Let us denote the potential function of the whole system by:
\[
U(\vec{x},t,T).
\]
	There is a cloud of free electrons, which are placed out of nodes 
and have as more as possible minimum total energy. 

	In the other words, the free electrons cloud
form a sea and the nodes are isolated islands on that sea.

	Let us assume that the potential function is constant on each node and there 
is a number $U_0(T)$ such that
\[
\bigcup_{n=1}^NE_n\subseteq \bigcup_{n=1}^NB(\vec{x}_n,r+\epsilon)\subseteq \{\vec{x};\,\,\,U_0(T)<U(\vec{x},t,T)\}.
\]

	Let there is a number $U_1(T)<U_0(T)$ such that all free 
electrons are placed in the subset 
\[
S(t,T)=\{\vec{x};\,\,\,U(\vec{x},t,T)<U_1(T)\}
\]
	We assume that the electrical conduction is 
related with the motion of free and semi - free electrons. 

	The forces keeping a semi - free electron on the boundary of node 
can not be electrical. Nevertheless, we put these conditions without 
any discussing.

	Now let us tray to explain the given above experimental results 
in frame of our model.

	At the enough low temperature the nodes get a perfect spherical shape.
As a consequence the subset 
\[
\{\vec{x};\,\,\,U_1(T)<U(\vec{x},t,T)\}
\]
become connected. On this subset the semi free electrons move without resistance.
So, the semi free electrons no restriction feel during of they motion in spite of 
potential function form. This is the cause of sudden rejection of resistance.

	Such a scenario depends upon the properties of the given metal of course. 
We assume that the spherical shape is characteristic for metals, which have the 
superconductivity property. 

	The critical temperature. The electrical forces push the electron on 
the boundary out of node. The electron stays on the boundary thanks of 
nucleus forces. If the electron's velocity is directed out of the node and 
it is large enough, the electron will go out from the node and it will go 
to the other node. 

	Since the segment, which connects the centers of the nodes, in low temperature? 
is placed out of $S(t,T)$, a semi - free electron will passes the distance between 
the nodes will not loss an energy. 
Thus, the segments which connect the centers of near placed nodes, form a 
ways by which the electrons can pass losing no energy.
This remark explains the existence of critical temperature.

	Now let as consider the effects before critical temperature. 
Let the temperature decreases. 
By weakening the chaotic motions inside of nodes, the shape of nodes 
become more like to the perfect ball. As a consequence the oscillation of 
potential function, on the lines connected the neighborhood placed 
nodes, become smaller. As a result the potential barriers arise.  
This follows that free - electrons must spend additional energy to overcome those barriers. 
Consequently, the resistance increases.

	It is well known that some ideal conductors are bad superconductors and vice-versa.
This effect has natural explanation in our model. Indeed conductivity is conditioned by free electrons while
superconductivity is conditioned by semi-free electrons.

	It is enough to note that in ideal conductors there are significantly more free electrons 
than semi-free electrons. 

	The role of magnetic field. On the electron, moving in a magnetic 
field, acts the force orthogonal to the direction of motion and to the 
direction of magnetic field.
 
	Let the magnetic field have the direction on OZ axe   and the OX axe 
connect the neighborhood placed colonies centers. Let an electron be 
on the boundary of the first colony. Then the electron will move by the 
curve, see [2], p. 172,
\begin{equation*}
\left(\frac{mv}{eH}\sin\frac{eHt}{m},\,\, \frac{mv}{eH}\left(1-\cos\frac{eHt}{m}\right),\,\,0\right)
\end{equation*}
Where $\vec{v}=(v,\,\,0,\,\,0)$ the electron has starting velocity.

	So, if  $H$ is bigger, the semi - free electron, which begins its motion 
on the surface of a nodes, will go out from the narrow way connecting the 
nodes centers. As a result electron appears in the cloud of free electrons. 
So, the conductor losses its superconducting property.
 
	Role of the pressure. For the following elements 

Be, Cr, Ba, Si, Ge,  Se, Sb, Te, Bi

the superconducting property was observed only in huge pressure, see also \cite{b:A}.
 
	For above mentioned metals, in low temperature, the nodes become the 
form of perfect balls, but they are far from. That is why the oscillation 
of potential function on the way connecting the nodes centers is bigger. 
If we increase the pressure, these nodes approached and the ways, by 
which the potential function has small oscillation, appear. As a result, 
the superconducting property arises.

	Role of current magnitude. If the current magnitude increases the 
semi - free electrons, moving in parallel ways, will interact and they 
will go out from the narrow ways which connect the nodes centers. 
As a consequence the conductor losses superconducting property.

\section{The guessed effects}

	Let us note that there are some effects which are caused of 
suggested new model.

	1. Let us note that the semi - free electron needs some energy to leave the node. 
So, if the current has small magnitude, then a semi - free electron
can not take part in conduction. 
This remark follows, that for currents of sufficient small magnitude, 
the superconductivity effect is absent. I have not on the hand an experimental result
conform this hypothesis.

	2. From our model it follows that superconductivity must be not homogeneous.
This is caused by configuration of narrow ways, which arise in low temperature 
and connect the nodes.
Let us note that those narrow ways makes a metal not homogeneous. 
Consequently, the currents in different 
directions will feel different resistances.

\end{document}